\documentclass[11pt]{article}
\usepackage{geometry}
\usepackage[utf8]{inputenc}
\usepackage{amsmath,amsthm,amsfonts,amssymb,mathtools,bbm}
\usepackage{booktabs}
\usepackage[framemethod=TikZ]{mdframed}
\usepackage[nomessages]{fp}
\usepackage{tikz, float, subcaption}
\usepackage{graphicx}
\usepackage{enumitem, fancyhdr, multicol}
\usepackage{stmaryrd}
\usepackage{complexity}
\usetikzlibrary{arrows, calc, fit, positioning, shapes.geometric, quantikz}
\usepackage{calculator}

\newcommand*{\nats}{\mathbb{N}}
\newcommand*{\reals}{\mathbb{R}}
\newcommand*{\complexes}{\mathbb{C}}

\newcommand*{\map}[3]{{{#1}:{#2}\rightarrow{#3}}}

\newcommand*{\st}{\;\middle\vert\;}
\newcommand*{\set}[1]{{\left\{#1\right\}}}
\newcommand*{\binary}{\set{0,1}}
\newcommand*{\bits}[1]{\binary^{#1}}

\newcommand*{\tup}[1]{\langle{} #1 \rangle{}}
\newcommand*{\braket}[2]{\left\langle{#1}\middle\vert{#2}\right\rangle}

\newcommand*{\wt}{\textup{wt}}

\newcommand*{\ctrlX}[1]{{\textup{C}_{#1}X}}
\newcommand*{\ctrlZ}[1]{{\textup{C}_{#1}Z}}
\newcommand*{\cnot}{\textup{C-NOT}}
\renewcommand*{\p}{\varphi}
\newcommand*{\cH}{\mathcal{H}}
\newcommand*{\cP}{\mathcal{P}}
\newcommand*{\cZ}{\mathcal{Z}}
\newcommand*{\mat}[1]{{\left[\begin{matrix}#1\end{matrix}\right]}}
\newcommand{\1}{\mathbf{1}}
\newcommand{\disjunion}{\mathop{\stackrel{.}{\cup}}}
\newcommand{\matelt}[3]{{\langle{#1}|{#2}|{#3}\rangle}}

\theoremstyle{plain}
\newtheorem{thm}{Theorem}
\newtheorem{conj}[thm]{Conjecture}

\newtheorem{lemma}[thm]{Lemma}

\theoremstyle{definition}
\newtheorem{defn}[thm]{Definition}

\newtheorem{rmrk}[thm]{Remark}

\title{Depth-$2$ $\QAC$ circuits cannot simulate quantum parity}
\author{Daniel Pad\'e\hspace{0.625in}Stephen Fenner\\University of South Carolina\thanks{Computer Science and Engineering Department, Columbia, SC 29208 USA\@.   \texttt{djpade@gmail.com}, \texttt{fenner.sa@gmail.com}.  Part of the work was done while the first author visited the fourth author in June and July, 2019.} \and Daniel Grier\\IQC\thanks{Institute for Quantum Computing, University of Waterloo, Waterloo, ON N2L3G1 Canada. \texttt{daniel.grier@uwaterloo.ca}} \and Thomas Thierauf\\Aalen University\thanks{Department of Computer Science and Electrical Engineering, Aalen, Germany.  \texttt{thomas.thierauf@uni-ulm.de}. Supported by DFG grant TH~472/5-1.}}
\date{\today}
\begin{document}
\maketitle

\begin{abstract}
We show that the quantum parity gate on $n>3$ qubits cannot be cleanly simulated by a quantum circuit with two layers of arbitrary C-SIGN gates of any arity and arbitrary $1$-qubit unitary gates, regardless of the number of allowed ancilla qubits.  This is the best known and first nontrivial separation between the parity gate and circuits of this form.
The same bounds also apply to the quantum fanout gate.  Our results are incomparable with those of Fang~et~al.~\cite{FFGHZ:fanout}, which apply to any constant depth but require a sublinear number of ancilla qubits on the simulating circuit.

\medskip

\noindent\textbf{Keywords:} quantum circuit, QAC, QACC, parity gate, fanout gate
\end{abstract}

\section{Introduction}
\label{sec:intro}

Quantum decoherence is a major obstacle to maintaining long quantum computations.  Large-scale quantum computers, if and when they are built, will very likely confront short decoherence times and so must act quickly to do useful computations.

A reasonable theoretical model of such computations are shallow quantum circuits, i.e., quantum circuits of small depth.  The decoherence dilemma has inspired much theoretical interest in the capabilities of these circuits, particularly circuits that have constant depth and polynomial size.  To solve useful problems, quantum circuits that are very shallow will require gates acting on several qubits at once. A major question then is this: do there exist multiple-qubit gates that are both potentially realizable and sufficient for powerful computation in small (even constant) depth?

It is known that, with the aid of \emph{fanout} gates (a certain multiqubit gate defined below), quantum circuits can do a variety of important tasks such as phase estimation and approximate Quantum Fourier Transform in essentially constant depth~\cite{HS:fanout}.  Are fanout gates necessary here?  If one only allows gates to act on $O(1)$ qubits each, it is clear that any decision problem computed by $o(\log n)$-depth quantum circuits with bounded error and  can only depend on $2^{o(\log n)}$ bits of the input (see \cite{FFGHZ:fanout} for a discussion).  Thus without allowing \emph{some} class of quantum gates with unbounded width (arity), no nontrivial decision problem can be computed by such a circuit.
What if we restrict to constant-width quantum gates, but we allow measurement of several qubits at the end, followed by post-processing by a polynomial-size classical circuit?  Here the situation is more complicated.  For certain types of constant-depth circuits---particularly, for circuits with constant-width gates followed by a classical AND applied to the measured results of all the output qubits---one can compute in polynomial time the result, provided there is a wide enough gap in the probabilities of getting a $0$-result versus a $1$-result~\cite{FGHZ:constant-depth}.  In contrast, Bravyi, Gosset, \& K\"{o}nig recently presented a search problem\footnote{In a search problem (or relation problem) there may be several possible acceptable outputs, and the device is only required to produce one of them.} that can be computed exactly by a constant-depth circuit with constant-width gates, and no classical probabilistic circuit of sublogarithmic depth can solve the same problem with high probability~\cite{BGK:quantum-advantage}.

Another type of multiqubit gate that has a natural definition is the quantum AND-gate, which flips the value of a target just when all the control qubits are on.\footnote{These gates are also called \emph{generalized Toffoli gates}.}  It is not clear whether such a gate will be easy to implement, but it is a natural question to compare the power of fanout versus quantum AND-gates with respect to constant-depth quantum computation.

A quantum circuit (actually a family of such circuits, one for each input size) using unbounded quantum AND-gates and single-qubit gates is called a \emph{$\QAC$ circuit}.  This is the quantum analogue of a classical $\AC$ circuit.    Takahashi \& Tani showed that the quantum AND-gate can be simulated exactly in constant depth by a quantum circuit with single-qubit gates and fanout gates~\cite{TT:constant-depth-collapse}.  The converse of the Takahashi \& Tani result---can a fanout gate be simulated exactly (or even approximately) by a constant-depth $\QAC$ circuit?---is still an open question, and is the main focus of this paper.  We conjecture that the answer is no, and our current results supply evidence in that direction, proving a separation between fanout and depth-2 $\QAC$ circuits.  It is known that quantum fanout gates are constant-depth equivalent to quantum parity gates~\cite{Moore:fanout}, and so the question at hand is a reasonable quantum analogue to the already proven separation between parity and $\AC^0$ in classical circuit complexity~\cite{Ajtai:AC0,FSS:AC0} (the superscript $0$ signifies constant-depth circuits).  This analogy is not perfect; in classical circuit complexity, fanout is usually taken for granted and used freely, and this is not the case with quantum circuits.

\begin{conj}
Constant-depth $\QAC$ circuits cannot simulate an unbounded quantum fanout gate.
\end{conj}

Partial progress on this conjecture
was made in~\cite{FFGHZ:fanout}, where it was shown that no constant-depth $\QAC$ circuit family (a.k.a.\ a $\QAC^0$ circuit family) \emph{with a sublinear number of ancilla qubits} can approximate a fanout gate.  Since then, progress on this conjecture has stalled until very recently.  In 2014, E.~Pius announced a result (slightly weaker than our main result) that parity (equivalently, fanout) of more than five qubits cannot be simulated cleanly by a $\QAC$ circuit with depth~2~\cite{Pius:QAC}.\footnote{We ignore single-qubit gates in determining the depth of a circuit, counting only those layers containing multiqubit gates.}  We have been unable to verify his proof completely.  Nonetheless, some ideas in that paper have been helpful in a new push to prove the conjecture.  We have recently found new techniques, described below, that go beyond those used in \cite{FFGHZ:fanout} to potentially prove that $\QAC$ circuits of any constant depth cannot simulate fanout gates.  Proving this conjecture would most likely lead to a separation of the corresponding language classes computed by these circuits: $\QAC^0 \ne \QACC^0$.  Here, $\QACC^0$ circuits are families of constant-depth circuits with single-qubit gates and unbounded mod-$q$ gates (for any $q>1$ constant across the circuits in the family).  Parity gates were shown to be depth-$1$ equivalent to fanout gates~\cite{Moore:fanout}, so these circuits are layer-for-layer equivalent to circuits with fanout gates instead, and it was shown in~\cite{GHMP:QACC} that mod-$q$ gates are simulatable by $\QAC$ circuits with parity gates in constant depth, and vice versa.

The main result of this paper (Theorem~\ref{thm:depth-2}, below) is that an $n$-qubit parity gate for $n\ge 4$ cannot be simulated cleanly by any depth-2 $\QAC$ circuit.  This result is tight in the sense that one can simulate the $3$-qubit parity gate with a depth-2 $\QAC$ circuit.

%
%

\section{Preliminaries}
\label{sec:prelims}


Following standard practice, we let $[n]$ denote the set $\{1,\ldots,n\}$ for any integer $n \ge 0$.  We write $z^*$ for the complex conjugate of a complex number $z$, and we write $A^*$ for the adjoint (Hermitian conjugate) of an operator $A$ on a Hilbert space.  Otherwise, our notation is fairly standard (see \cite{KLM:quantum-book,KSV:quantum-book,NC:quantumbook} for example).

For $n\ge 0$ and $s \in \bits{n}$, we let $\wt(s)$ denote the Hamming weight of $s$, and we let $\oplus s\in\binary$ denote the parity of the bits of $s$, i.e., $\oplus s := \wt(s) \bmod 2$.

For $m\ge 0$, we let $\cH_m$ denote the Hilbert space on $m$ qubits, labeled $1,\ldots,m$.  Thus $\cH_m$ has dimension $2^m$, and is isomorphic to $\left(\complexes^{2}\right)^{\otimes m}$ via the usual computational basis.
If $S$ is some subset of $[m]$, then we let $\cH_S$ denote the Hilbert space of the qubits with labels in $S$.  Thus for example, $\cH_m = \cH_{[m]}$.  For disjoint $S,T\subseteq [m]$, there is a natural isomorphism $\cH_{S\cup T} \cong \cH_S \otimes \cH_T$, and so we will not distinguish between these.  For $S\subseteq [m]$, we let $\overline{S}$ denote $[m]\setminus S$.

Our quantum circuit model with unitary gates is standard, found in several textbooks, including~\cite{NC:quantumbook,KLM:quantum-book}.  We assume our circuit acts on $\cH_m$ for some $m\in\nats$.  We assume qubits $1,\ldots,n$ are the \emph{input qubits}, for some $n\le m$, and the rest are \emph{ancilla qubits}.  For any single-qubit unitary operator $U$, we let $U_i$ denote $U$ acting on qubit~$i$, for $1\le i\le m$.  (Note that $U_i$ acts on the entire space of $m$ qubits; it is the tensor product of $U$ with the identity operator $I$ acting on the rest of the qubits.)

All the quantum circuits circuits we consider are allowed arbitrary single-qubit gates.  These gates do not count toward the depth of the circuit; only layers of multiqubit gates are counted for the depth.  For example, a depth-$1$ circuit many have multiqubit gates acting on disjoint set of qubits simultanously (in a single layer), preceded and followed on each qubit with an arbitrary single-qubit gate.

The $1$-qubit Pauli gates are defined as usual:
\begin{align*}
X &:= \mat{0&1\\1&0}\;, & Y &:= \mat{0&-i\\i&0}\;, & Z &:= \mat{1&0\\0&-1}\;.
\end{align*}

The \emph{$k$-qubit fanout gate} $F_k$ acts on $k\ge 2$ qubits, where one qubit, the first, say, is the \emph{control} and the rest are targets:
\[ F_k\ket{x_1,x_2,\cdots,x_k} = \ket{x_1,x_1\oplus x_2,\cdots,x_1\oplus x_k} \]
for all $x_1,\ldots,x_k\in\{0,1\}$.  $F_k$ is equivalent to applying $k-1$ many $\cnot$ gates in succession, all with the same control qubit, and targets $2$ through $k$, respectively.  If the targets are initially all in the $\ket{0}$ state, then $F_k$ copies the classical value of the control qubit to each of the targets.\footnote{This does not violate the no-cloning theorem, because only the classical value is copied.}

The \emph{$k$-qubit parity gate} $\oplus_k$ acts on $k\ge 2$ qubits, where the first (say) is the target and the rest are control qubits:
\[ \oplus_k\ket{x_1,x_2,\ldots,x_k} = \ket{x_1\oplus \cdots\oplus x_k,x_2,\ldots,x_k} \]
for any $x_1,\ldots,x_k\in\{0,1\}$.  The parity gate $\oplus_k$ results from $F_k$ by conjugating each qubit with a Hadamard gate $H$~\cite{Moore:fanout}, that is,
\[ \oplus_k = (H_1 H_2\cdots H_k) F_k (H_1 H_2\cdots H_k) \]
and vice versa.

The \emph{$k$-qubit quantum AND-gate} (a.k.a.\ the generalized Toffoli gate) $\ctrlX{k}$ flips the value of the target (the first qubit, say) just when all control bits are $1$:
\[ \ctrlX{k}\ket{x_1,x_2,\ldots,x_k} = \ket{x_1\oplus (x_2\cdots x_k),x_2,\ldots,x_k} \]
for any $x_1,\ldots,x_k\in\{0,1\}$.  For example $\ctrlX{2} = F_2 = \cnot$.

The gates mentioned above are all ``classical'' in the sense that they map basis states to basis states.  This is not true of the C-SIGN gate.

The \emph{$k$-qubit C-SIGN gate} $\ctrlZ{k}$ flips the overall phase just when all bits are $1$:
\[ \ctrlZ{k}\ket{x_1,\ldots,x_k} = (-1)^{x_1\cdots x_k}\ket{x_1,\ldots,x_k} \]
for any $x_1,\ldots,x_k\in\{0,1\}$.  The C-SIGN gate results from the quantum AND-gate by conjugating the target qubit with Hadamard gates:
\[ \ctrlZ{k} = H_1\ctrlX{k} H_1 \]
and vice versa:
\[ \ctrlX{k} = H_1\ctrlZ{k} H_1\;. \]
A technical advantage of the C-SIGN gate over the quantum AND-gate is that the C-SIGN gate has no distinguished target or control qubits; all qubits incident to the gate are on the ``same footing;'' more precisely, the C-SIGN gate commutes with the SWAP operator applied to any pair of its qubits.  With that in mind we define, for any subset $S$ of the qubits of a multiqubit register, the gate $\ctrlZ{S}$ as the C-SIGN gate acting on the qubits in $S$.  Note, however, that $\ctrlZ{S}$ is a unitary operator on the entire register, being the tensor product of a C-SIGN gate on the qubits in $S$ with the identity operator on the other qubits.  We define $\ctrlZ{\emptyset} := -I$ by convention, where $I$ is the identity operator on the register.  We also refer to a C-SIGN gate acting on an unspecified set of qubits as a $\ctrlZ{}$ gate.

\begin{defn}
A \emph{$\QAC$ circuit} is a quantum circuit that includes $\ctrlZ{}$ gates and (arbitrary) single-qubit gates.  For $\QAC$ circuit $C$, we define the \emph{depth} of $C$ in the standard way, except we do not include single-qubit gates as contributing to the depth, i.e., as if all single-qubit gates are removed.
\end{defn}

\begin{defn}
If $G$ is an $n$-qubit unitary operator and $C$ is a quantum circuit on $m\ge n$ qubits, we say that \emph{$C$ cleanly simulates $G$} if, for all $x\in\{0,1\}^n$,
\[ C(\ket{x}\otimes\ket{0^{m-n}}) = (G\ket{x})\otimes\ket{0^{m-n}}\;. \]
\end{defn}

So particularly, when the ancilla qubits are initially all $0$, they are returned to being all $0$ at the end.

In this paper we prove the following theorem:

\begin{thm}\label{thm:depth-2}
No depth-$2$ $\QAC$ circuit cleanly simulates $\oplus_n$ for any $n\ge 4$, regardless of the number of its ancilla qubits.
\end{thm}

This result is tight in the sense that there is a simple $3$-qubit depth-$2$ $\QAC$ circuit that cleanly simulates $\oplus_3$.  Theorem~\ref{thm:depth-2} improves upon Pius's announced result above by reducing the number of input qubits.

To prove this theorem, we introduce a new technique that has promise for increasing the depth hypothesis well beyond~$2$.  Our technique makes use of a specific entangling property of the C-SIGN gate.  Roughly, any essential application of a C-SIGN gate leaves all its qubits entangled, provided they were not so entangled to begin with.  By ``essential'' we mean that the gate does not disappear or simplify to a gate of smaller arity.

\begin{defn}\label{def:S-separable}
Suppose we have an $n$-qubit register with qubits labeled $1,\ldots,n$.  Let $\ket{\psi}$ be some state of the $n$~qubits, and let $S$ be a subset of the qubits with $|S| \ge 2$.  We say that $\ket{\psi}$ is \emph{$S$-separable} if there exists a bipartition of $[n]$ into sets $A$ and $B$ such that $A\cap S \ne \emptyset$, $B\cap S\ne \emptyset$, and $\ket{\psi} = \ket{\psi}_A\otimes\ket{\psi}_B$ for two states $\ket{\psi}_A$ and $\ket{\psi}_B$ on the qubits in $A$ and in $B$, respectively.  If $\ket{\psi}$ is not $S$-separable, we say it is \emph{$S$-entangled}.
\end{defn}

\begin{defn}\label{def:simplify}
Suppose we have an $n$-qubit register with qubits labeled $1,\ldots,n$, a set $S\subseteq [n]$, and an $n$-qubit state $\ket{\psi}$.  We say that $\ctrlZ{S}$ \emph{simplifies} on $\ket{\psi}$ if either (a) $\ctrlZ{S}\ket{\psi} = \ket{\psi}$ or (b) $\ctrlZ{S}\ket{\psi} = \ctrlZ{T}\ket{\psi} \ne \ket{\psi}$ for some proper subset $T\subset S$.  In case (a), we say that \emph{$\ctrlZ{S}$ disappears (or is turned off) on $\ket{\psi}$}; in case (b), we say that \emph{$\ctrlZ{S}$ simplifies to $\ctrlZ{T}$ on $\ket{\psi}$}.
\end{defn}

Observe that the two cases (a) and (b) in Definition~\ref{def:simplify} above are mutually exclusive, given $S$ and $\ket{\psi}$.  Also observe that $\ctrlZ{S}$ disappears on $\ket{\psi}$ if and only if $\braket{x}{\psi} = 0$ for every computational basis state $\ket{x}$ such that the string $x$ has $1$'s in all positions in $S$.  $\ctrlZ{S}$ simplifies to $\ctrlZ{T}$ on $\ket{\psi}$ if and only if $\braket{x}{\psi} = 0$ for every computational basis state $\ket{x}$ where $x$ has a $0$ in some position in $S-T$; equivalently, $\ket{\psi}$ factors into a tensor product of a $\ket{1}$ state of each qubit in $S-T$, along with some arbitrary state of the rest of the qubits.

In Appendix~\ref{sec:S-entangled} we prove the following lemma:

\begin{lemma}[Entanglement Lemma]\label{lem:S-entangled}
Suppose we have an $n$-qubit register as in Definition~\ref{def:simplify}, and let $S$ be a subset of $[n]$.  Let $\ket{\psi}$ be any state of the register, and let $\ket{\p} := \ctrlZ{S}\ket{\psi}$.  Then at least one of the following must hold: (1) $\ket{\psi}$ is $S$-entangled; (2) $\ket{\p}$ is $S$-entangled; (3) $\ctrlZ{S}$ simplifies on $\ket{\psi}$.
\end{lemma}

%

\begin{defn}\label{def:parity-space}
Given $n\ge 1$ and $b\in\binary$, we define the subspace $\cP_b$ of $\cH_n$ to be the space spanned by $\set{\ket{x} \st x\in\bits{n} \wedge \oplus x = b}$.
\end{defn}

Clearly, $\dim\cP_0 = \dim\cP_1 = 2^{n-1}$, and $\cH_n$ is the direct sum of $\cP_0$ and $\cP_1$.

\begin{defn}[Parity of a State]
  Given an $n$-qubit state $\ket{\psi}\in\cH_n$ and $b\in\binary$,
we say that \emph{$\ket{\psi}$ has pure parity $b$} if $\ket{\psi}\in\cP_b$.  We say that $\ket{\psi}$ is a \emph{pure parity state} if $\ket{\psi}$ has pure parity $b$ for some $b\in\binary$.  Otherwise, we say that $\ket{\psi}$ has \emph{mixed parity}.
\end{defn}


\begin{defn}\label{def:weakly-compute}
Let $n \ge 1$.  A quantum circuit $C$ \emph{weakly computes $\parity_n$} if $C$ acts on $m$ qubits, for some $m\ge n$, and there exists state $\ket{\psi}\in\complexes^{2^{m-n}}$ such that, for any $x\in\bits{n}$, there exists state $\ket{\p_x}\in\complexes^{2^{m-1}}$ such that
  \begin{align*}
    C(\ket{x} \otimes \ket{\psi}) = \ket{\parity x} \otimes \ket{\p_x} \;.
  \end{align*}
\end{defn}

In the circuit $C$ above, we consider the first qubit to be both the target and an input qubit.  The $m-n$ non-input qubits are ancilla qubits.  Clearly, if a circuit cleanly simulates $\parity_n$, then it weakly computes $\parity_n$.

\begin{lemma}\label{lem:kill-parity}
Given any $n$-qubit unitary operators $U_1,\ldots,U_k$ for some $k < 2^{n-1}$ and any bit $b\in\binary$, there is an $n$-qubit state $\ket{\psi}$ with pure parity~$b$ such that $\bra{1^n} U_iU_{i-1}\cdots U_1 \ket{\psi} = 0$ for all $1\le i\le k$.
\end{lemma}
\begin{proof}
Let $\cP_0$ and $\cP_1$ be as in Definition~\ref{def:parity-space}.  For $1\le i\le k$, set $V_i := U_i\cdots U_1$, and let $\cZ_i\subseteq\cH_n$ be the $(2^n-1)$-dimensional subspace of $\cH_n$ spanned by $\set{V_i^*\ket{x}: x\in\binary^n\setminus \{1^n\}}$.  Then for all $i$, $\bra{1^n} V_i \ket{\psi} = 0$
for any state $\ket{\psi}\in\cZ_i$.  Letting $\cZ := \bigcap_{i=1}^k \cZ_i$, we see that $\dim(\cZ) \ge 2^n-k$.  For $b\in\binary$, we then have
\[ \dim(\cP_b\cap\cZ) = \dim\cP_b + \dim\cZ - \dim(\cP_b+\cZ) \ge \dim\cP_b + \dim\cZ - 2^n \ge 2^{n-1} + (2^n - k) - 2^n \ge 1. \]
It follows that we can choose a state (unit vector) $\ket{\psi}$ in $\cP_b\cap\cZ$, and this vector has the desired properties.
\end{proof}

\section{Lower Bounds}
\subsection{Depth-1 circuits}

\begin{lemma}\label{lem:single-layer}
There is no depth-$1$ $\QAC$ circuit that weakly computes $\parity_n$ for $n\ge 3$.
\end{lemma}
\begin{proof}
Consider such a circuit $C$ on at least three input qubits.  These must all be incident to a single $\ctrlZ{S}$ gate for some $S\supseteq\{1,2,3\}$, for otherwise there is a non-target input qubit that does not interact with the target qubit at all, whence $C$ cannot weakly compute $\parity_n$.  Then by Lemma~\ref{lem:kill-parity}, input qubits~$1$ and $2$ (including the target) can be initially committed to a pure-parity state $\ket{\psi}$ that turns off $\ctrlZ{S}$. Then given the initial state $\ket{\psi}\otimes\cdots$, input qubit~$3$ does not affect the target qubit.  This is a contradiction, because toggling qubit~$3$'s initial state between $\ket{0}$ and $\ket{1}$ while qubits~$1$ and $2$ are in state $\ket{\psi}$ changes the parity of the inputs and so must flip the value of the target on the output.
\end{proof}


\subsection{Depth-2 circuits}

A depth-$d$ circuit can have $d$ layers of $\ctrlZ{}$ gates, which we call \emph{layers~$1$ through $d$}, respectively, layer~$1$ lying to the left of layer~$2$, etc.  To the left, right, and in between these layers are arbitrary $1$-qubit gates.  Viewing the circuit as acting from left to right, the leftmost $1$-qubit gates are applied first; we say these gates are on layer~$0.5$.  Then the layer-$1$ $\ctrlZ{}$ gates are applied, followed by the $1$-qubit gates between layers~$1$ and $2$ (layer~$1.5$), followed by the $\ctrlZ{}$ gates on layer~$2$, and so on, then finally the rightmost layer of $1$-qubit gates (layer~$d+\frac{1}{2}$).

\begin{defn}
A single-qubit gate is \emph{mixing} if, in its matrix representation with respect to the computational basis, all entries are non-zero.
\end{defn}

Observe that a $1$-qubit unitary gate $U$ is mixing if and only if $U^*$ is mixing.

\begin{lemma}\label{lem:comp-to-comp}
Let $G$ be a single-qubit gate.  $G$ is non-mixing if and only if $G$ applied to any computational basis state outputs a computational basis state up to a phase, i.e., for any $b\in\binary$ there exist $c\in\binary$ and $\eta\in\reals$ such that $G \ket{b} = e^{i\eta}\ket{c}$.  Moreover, if this is the case, then either $G = e^{i\beta}\,e^{i\alpha Z}$ or $G = e^{i\beta}\,X\,e^{i\alpha Z}$ for some $\alpha,\beta\in\reals$.
\end{lemma}
\begin{proof}
If $G$ is non-mixing, then due to the normalization of the rows and columns of any $2\times 2$ unitary matrix, $G$ can be written in one of these forms, for some $\theta,\phi\in\reals$:
  \begin{align*}
    \mat{e^{i\theta}&0\\0&e^{i\phi}}
            &= e^{i\beta}\,e^{i\alpha Z} &&\text{or} &
    \mat{0&e^{i\phi}\\e^{i\theta}&0} &= e^{i\beta}\,X\,e^{i\alpha Z}\;,
  \end{align*}
where $\alpha = (\theta-\phi)/2$ and $\beta = (\theta+\phi)/2$.  Applying either of these matrices to a computational basis state yields a computational basis state up to a phase.

The reverse implication is obvious.
\end{proof}

\begin{defn}
In a depth-$d$ $\QAC$ circuit, if a qubit $q$ encounters a non-mixing $1$-qubit gate in layer~$d+\frac{1}{2}$, then we say that $q$ is \emph{pass-through}.  If $q$ encounters a non-mixing $1$-qubit gate in layer $\frac{1}{2}$, then we say that $q$ is \emph{pass-in}.
\end{defn}

\begin{lemma}\label{lem:target-cant-pass}
For any $n\ge 1$ and $d\ge 2$, let $C$ be a depth-$d$ $\QAC$ circuit that weakly computes $\parity_n$.  If $C$'s target is either pass-through or does not encounter a $\ctrlZ{}$ gate on layer~$d$, then there exists a depth-$(d-1)$ $\QAC$ circuit that weakly computes $\parity_n$ with the same initial ancilla state as $C$.
\end{lemma}

\begin{proof}
Fix an initial ancilla state $\ket{\psi}$ that witnesses $C$ weakly computing $\parity_n$.  By Lemma~\ref{lem:comp-to-comp}, for any classical input $x$ combined with $\ket{\psi}$, the target (qubit~$1$) is in an unentangled computational basis state $\ket{b}$ at layer~$d$ (where $b\in\{0,1\}$ depends on $x$).  Thus a layer-$d$ $\ctrlZ{}$ gate (if there is one) acting on the target either disappears or simplifies to a $\ctrlZ{}$ gate not acting on the target, depending on $b$.  In either case, the (unentangled) state of the target is unchanged across layer~$d$.  Let $C'$ be the depth-$(d-1)$ circuit obtained from $C$ by removing all gates on layer~$d$, removing all non-target gates on layer~$d+\frac{1}{2}$, and combining the target gate on layer~$d+\frac{1}{2}$ (if any) with the target gate on layer~$d-\frac{1}{2}$.  For any classical input, the final state of the target is thus the same with $C'$ as with $C$, and so $C'$ weakly computes $\parity_n$ with initial ancilla state $\ket{\psi}$.
\end{proof}

The following lemma is a corollary to Lemma~\ref{lem:target-cant-pass}.

\begin{lemma}\label{lem:target-cant-pass-2}
In any depth-$2$ $\QAC$ circuit weakly computing $\parity_n$ for $n\ge 3$, there is a $\ctrlZ{}$ gate on layer~$2$ acting on the target, and the target is not pass-through.
\end{lemma}

\begin{proof}
By Lemmas~\ref{lem:single-layer} and \ref{lem:target-cant-pass}.
\end{proof}

\begin{lemma}\label{lem:shall-not-pass}
In a depth-$2$ $\QAC$ circuit $C$, if any non-target input qubit sharing a layer-$2$ $\ctrlZ{}$ gate with the target is pass-through, or if any non-target input qubit sharing a layer-$1$ $\ctrlZ{}$ gate with the target is pass-in, then $C$ cannot simulate $\parity_n$ cleanly for $n>3$.
\end{lemma}
\begin{proof}
Suppose such a $C$ cleanly simulates $\parity_n$ for some $n$, and first consider any non-target input qubit $q$ that is not pass-through but shares a layer-$2$ $\ctrlZ{}$ gate with the target.  By Lemma~\ref{lem:comp-to-comp}, the initial state of $q$ can be committed to either $\ket{0}$ or $\ket{1}$ such that $q$ is in state $\ket{0}$ across the layer-$2$ $\ctrlZ{}$ gate, turning that gate off.  Treating $q$ as an ancilla qubit, $C$ is now equivalent to a depth-$1$ circuit weakly computing $\parity_{n-1}$, which by Lemma~\ref{lem:single-layer} can only weakly compute parity on at most $2$ qubits.  Thus, $n\le 3$.

Since the parity gate is its own inverse, $C$ cleanly simulates parity if and only if the inverse $C^*$ of $C$ cleanly simulates parity.  Thus we can apply the whole argument of the last paragraph to the inverse of $C$---a ``mirror image'' argument---showing that if $q$ is not pass-in but shares a layer-$1$ $\ctrlZ{}$ gate with the target in circuit $C$, then $C^*$ cannot cleanly simulate $\parity_n$ for $n>3$, and thus neither can $C$.
\end{proof}


\begin{lemma}\label{lem:kill-parity-2}
Consider a depth-$2$ $\QAC$ circuit cleanly simulating $\parity_n$ for $n\ge 3$.  For any three input qubits $q_1$, $q_2$, and $q_3$ that share a common $\ctrlZ{}$ gate on both layers~$1$ and $2$ (possibly with other qubits), there exists a $3$-qubit pure-parity input state of $q_1,q_2,q_3$ that turns off both $\ctrlZ{}$ gates.
\end{lemma}

\begin{proof}
Let $\ctrlZ{}_1$ and $\ctrlZ{}_2$ be the $\ctrlZ{}$ gates shared by $q_1,q_2,q_3$ on layers~$1$ and $2$, respectively.  We apply Lemma~\ref{lem:kill-parity} for $n=3$ and $k=2$, where $U_1$ is the tensor product of the three $1$-qubit gates on qubits $q_1$, $q_2$, and $q_3$ in layer $0.5$, and $U_2$ is the same except on layer $1.5$.  We have $2 = k < 4 = 2^{n-1}$, so by Lemma~\ref{lem:kill-parity} there exists a $3$-qubit state $\ket{\psi}$ on $q_1,q_2,q_3$ such that $\bra{111}U_1\ket{\psi} = \bra{111}U_2U_1\ket{\psi} = 0$.  We see that $\ket{\psi_1} := U_1\ket{\psi}$ is the state of the $3$ qubits just prior to layer~$1$, as depicted in Figure~\ref{fig:3-qubit-CZ}.
  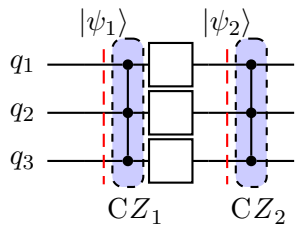
\begin{figure}[H]
    \centering
    \begin{tikzcd}
      \lstick{$q_{1}$} & \qw{}\slice{$\ket{\psi_{1}}$} & \ctrl{2}\gategroup[3,steps=1,style={dashed,rounded corners,inner xsep=0pt,fill=blue!20},label style={label position=below,anchor=north,yshift=-0.2cm,xshift=0.1cm},background]{$\ctrlZ{}_{1}$}  & \gate{} & \qw{} \slice{$\ket{\psi_{2}}$} & \ctrl{2}\gategroup[3,steps=1,style={dashed,rounded corners,inner xsep=0pt,fill=blue!20},label style={label position=below,anchor=north,yshift=-0.2cm,xshift=0.1cm},background]{$\ctrlZ{}_{2}$}   & \qw{} \\
      \lstick{$q_{2}$} & \qw{}                         & \control{} & \gate{} & \qw{} & \control{}  & \qw{} \\
      \lstick{$q_{3}$} & \qw{}                         & \control{} & \gate{} & \qw{} & \control{}  & \qw{}
    \end{tikzcd}
    \caption{Simplified depth-$2$ circuit, ignoring single-qubit gates on the input and output.}\label{fig:3-qubit-CZ}
  \end{figure}
Since $\tup{111|\psi_1} = 0$, the state $\ket{\psi_1}$ turns off the gate $\ctrlZ{}_1$, whence the state just prior to layer~$2$ is $\ket{\psi_2} := U_2U_1\ket{\psi}$.  Again by Lemma~\ref{lem:kill-parity}, $\ket{\psi_2}$ turns off $\ctrlZ{}_2$ on layer~$2$.
\end{proof}


We are now ready to prove our main result.

\begin{thm}
There is no depth-$2$ $\QAC$ circuit cleanly simulating $\parity_n$ for $n>3$. \end{thm}

\begin{proof}
Suppose $C$ is a depth-$2$ $\QAC$ circuit cleanly simulating $\parity_n$ for $n > 3$.  By Lemma~\ref{lem:target-cant-pass-2} applied to $C$ and Lemma~\ref{lem:shall-not-pass}, there must be a $\ctrlZ{}$ gate $G$ in layer~$2$ acting on the target, and none of the input qubits $G$ acts on are pass-through.  By Lemma~\ref{lem:target-cant-pass-2} applied to the inverse circuit $C^*$ (i.e., the mirror image argument) and Lemma~\ref{lem:shall-not-pass}, there is a $\ctrlZ{}$ gate $U$ in layer~$1$ acting on the target, and none of the input qubits $U$ acts on are pass-in.  Let $S$ be the set of qubits acted on by $G$ (so $G = \ctrlZ{S}$), and let $T$ be the set of qubits acted on by $U$, noting that both $S$ and $T$ include the target.  We can assume as well that none of the ancilla qubits in $S$ are pass-through; otherwise, either $G$ disappears for all classical inputs or $G$ simplifies to the same proper subset of $S$ for all classical inputs; in the former case, $C$ is equivalent to a depth-$1$ $\QAC$ circuit weakly computing $\parity_n$, which is impossible by Lemma~\ref{lem:single-layer}, and in the latter case, $G$ can be replaced with a $\ctrlZ{}$ gate of smaller arity that does not include the pass-through ancilla qubits (but still acts on the target) to obtain an equivalent circuit.  By a similar mirror argument, we can assume that none of the qubits in $T$ are pass-in.

By cleanliness and the fact that all $1$-qubit gates on layer~$2.5$ acting on gates in $S$ are mixing, for any classical input $x$, the state of the qubits in $S$ just after layer~$2$ is a tensor product of $1$-qubit states that are all nontrivial superpositions of computational basis states.  It follows that $G$ does not simplify on any layer~$2$ state arising from a classical input.

Now by the entanglement lemma (Lemma~\ref{lem:S-entangled}), on any classical input $x$, the state $\ket{\psi_x}$ just before layer~$2$ must be $S$-entangled, since the state after layer~$2$ is $S$-separable.  Since the single-qubit gates on layer~$1.5$ do not affect $S$-entanglement, the state $\ket{\p_x}$ just after layer~$1$ is also $S$-entangled.  Since the state immediately before layer~$1$ is clearly $S$-separable, it must be that all the qubits in $S$ must be acted upon by $U$.  This implies that $U$ must act on all input qubits; otherwise, there exists an input qubit that is acted upon neither by $U$ nor by $G$ and is thus not connected to the target at all.  We thus have that $S\subseteq T$ and $T$ includes all input qubits.  By the mirror argument, we get that $T\subset S$ as well; $U$ does not simplify for any classical input, because none of its qubits is pass-in, and so after $U$ is applied, the state is $T$-entangled and stays $T$-entangled up to layer~$2$, requiring $G$ to act on all the qubits in $T$ since there is no entanglement after layer~$2$.  Thus we have $S=T$, from which it follows that $G$ acts on all input qubits.

Finally, let $q_1$, $q_2$, and $q_3$ be any three input qubits, one of which is the target.  These three are all acted on by both $U$ and $G$.  Since $n>3$, there is at least one remaining (non-target) input qubit $q_4$.  By Lemma~\ref{lem:kill-parity-2}, there exists a pure parity-$0$ state $\ket{\psi}$ on $q_1,q_2,q_3$ that turns off both both $U$ and $G$.  With $q_1,q_2,q_3$ initially in this state, $q_4$ is not connected to the target, and thus cannot influence the final state of the target at all.  This contradicts the fact that the parity depends on all input qubits.
\end{proof}

\subsection{Further Research}

Our techniques currently work for depth~$2$, but obviously, we would like to prove limitations on $\QAC$ circuits of higher depth.  The entanglement lemma (Lemma~\ref{lem:S-entangled}) is stronger than needed for the current result; a weaker form, which assumes that $\ket{\psi}$ is factorable into single-qubit states, is easier to prove and still adequate for the current results.  We hope the stronger version will be useful for depth~$3$ and beyond, however.
Lemma~\ref{lem:kill-parity} is also stronger than needed for the current results; by committing clusters of input qubits to certain states, we can turn off C-SIGN gates through more than two layers.  These two lemmas provide powerful tools for dealing with $\QAC$ circuits of higher depth.  By simplifying a circuit in the right way, one can reduce its effective depth, and this in turn can lead to an inductive proof of the limitations of such circuits.

More specifically, Lemma~\ref{lem:S-entangled} may be useful for depth~$3$ and beyond because it disallows many different circuit topologies for $\QAC$ circuits computing parity.  For example, the following circuit topology is impossible for computing parity (or any classical reversible function for that matter) cleanly unless the middle gate simplifies:
\begin{center}
\begin{picture}(0,0)%
\includegraphics{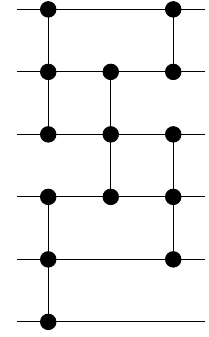}%
\end{picture}%
\setlength{\unitlength}{1973sp}%
\begingroup\makeatletter\ifx\SetFigFont\undefined%
\gdef\SetFigFont#1#2#3#4#5{%
  \reset@font\fontsize{#1}{#2pt}%
  \fontfamily{#3}\fontseries{#4}\fontshape{#5}%
  \selectfont}%
\fi\endgroup%
\begin{picture}(1977,3228)(2836,-4405)
\put(2851,-4336){\makebox(0,0)[rb]{\smash{{\SetFigFont{6}{7.2}{\rmdefault}{\mddefault}{\updefault}$6$}}}}
\put(2851,-1336){\makebox(0,0)[rb]{\smash{{\SetFigFont{6}{7.2}{\rmdefault}{\mddefault}{\updefault}$1$}}}}
\put(2851,-1936){\makebox(0,0)[rb]{\smash{{\SetFigFont{6}{7.2}{\rmdefault}{\mddefault}{\updefault}$2$}}}}
\put(2851,-2536){\makebox(0,0)[rb]{\smash{{\SetFigFont{6}{7.2}{\rmdefault}{\mddefault}{\updefault}$3$}}}}
\put(2851,-3136){\makebox(0,0)[rb]{\smash{{\SetFigFont{6}{7.2}{\rmdefault}{\mddefault}{\updefault}$4$}}}}
\put(2851,-3736){\makebox(0,0)[rb]{\smash{{\SetFigFont{6}{7.2}{\rmdefault}{\mddefault}{\updefault}$5$}}}}
\end{picture}%

\end{center}
(Here only the C-SIGN gates are shown; the single qubit gates are suppressed.)
The reason is that, for any classical input, the state on the far left is completely separable, and so the state immediately after the first layer is $\{2,3,4\}$-separable (via the partition $A = \{1,2,3\}$ and $B = \{4,5,6\}$).  If the middle gate does not simplify, then by the lemma, the state $\ket{\psi}$ immediately to its right must be $\{2,3,4\}$-entangled.  Now assuming a clean simulation, the state on the far right is completely separable, and so running the circuit backwards from the right, we see that $\ket{\psi}$ must be $\{2,3,4\}$-separable (via the partition $A = \{1,2\}$ and $B = \{3,4,5,6\}$).  Noting that single-qubit gates do not affect the $S$-separability of any state, this is a contradiction.

We note that the techniques used to prove that parity cannot be computed by classical $\AC^0$ circuits (i.e., random restrictions and switching lemmas) are not necessarily needed or even relevant here, because fanout is taken for granted in the classical case, unlike in the quantum case.

Finally, we only consider exact simulations in this paper.  A natural question to ask is whether one can prove nonapproximability results as well.  We suspect some of our techniques---e.g., the entanglement lemma---can be strengthened to help with some of these results, but new techniques will certainly also be needed.

\bibliographystyle{plain}
\bibliography{../../research/bib/master}

\begin{thebibliography}{10}

\bibitem{Ajtai:AC0}
M.~Ajtai.
\newblock {$\Sigma^1_1$} formul\ae\ on finite structures.
\newblock {\em Annals of Pure and Applied Logic}, 24:1--48, 1983.

\bibitem{BGK:quantum-advantage}
S.~Bravyi, D.~Gosset, and R.~K\"{o}nig.
\newblock Quantum advantage with shallow circuits.
\newblock {\em Science}, 362(6412):308--311, 2018.

\bibitem{FFGHZ:fanout}
M.~Fang, S.~Fenner, F.~Green, S.~Homer, and Y.~Zhang.
\newblock Quantum lower bounds for fanout.
\newblock {\em Quantum Information and Computation}, 6:46--57, 2006.

\bibitem{FGHZ:constant-depth}
S.~Fenner, F.~Green, S.~Homer, and Y.~Zhang.
\newblock Bounds on the power of constant-depth quantum circuits.
\newblock In {\em Proceedings of the 15th International Symposium on
  Fundamentals of Computation Theory}, volume 3623 of {\em Lecture Notes in
  Computer Science}, pages 44--55. Springer-Verlag, 2005.

\bibitem{FSS:AC0}
M.~Furst, J.~B. Saxe, and M.~Sipser.
\newblock Parity, circuits, and the polynomial time hierarchy.
\newblock {\em Mathematical Systems Theory}, 17:13--27, 1984.

\bibitem{GHMP:QACC}
F.~Green, S.~Homer, C.~Moore, and C.~Pollett.
\newblock Counting, fanout and the complexity of quantum {ACC}.
\newblock {\em Quantum Information and Computation}, 2:35--65, 2002.

\bibitem{HS:fanout}
P.~H{\o}yer and R.~{\v{S}}palek.
\newblock Quantum circuits with unbounded fan-out.
\newblock In {\em Proceedings of the 20th Symposium on Theoretical Aspects of
  Computer Science}, volume 2607 of {\em Lecture Notes in Computer Science},
  pages 234--246. Springer-Verlag, 2003.

\bibitem{KLM:quantum-book}
P.~Kaye, R.~Laflamme, and M.~Mosca.
\newblock {\em An Introduction to Quantum Computing}.
\newblock Oxford University Press, 2007.

\bibitem{KSV:quantum-book}
A.~Yu.\ Kitaev, A.~H. Shen, and M.~N. Vyalyi.
\newblock {\em Classical and quantum computation}.
\newblock American Mathematical Society, Providence, RI, 2002.

\bibitem{Moore:fanout}
C.~Moore.
\newblock Quantum circuits: Fanout, parity, and counting, 1999.
\newblock Manuscript.

\bibitem{NC:quantumbook}
M.~A. Nielsen and I.~L. Chuang.
\newblock {\em Quantum Computation and Quantum Information}.
\newblock Cambridge University Press, 2000.

\bibitem{Pius:QAC}
Einar Pius.
\newblock {\em Parallel Quantum Computing From Theory to Practice}.
\newblock PhD thesis, The University of Edinburgh, 8 2014.

\bibitem{TT:constant-depth-collapse}
Y.~Takahashi and S.~Tani.
\newblock Collapse of the hierarchy of constant-depth exact quantum circuits.
\newblock {\em Computational Complexity}, 25(4):849--881, 2016.
\newblock Conference version in Proceedings of the 28th IEEE Conference on
  Computational Complexity (CCC 2013).

\end{thebibliography}

\newpage

\appendix

\section{Proof of the Entanglement Lemma}
\label{sec:S-entangled}

Here we prove the entanglement lemma (Lemma~\ref{lem:S-entangled}) in a slightly more general context.

Recall that the C-SIGN gate on $k$ qubits is a unitary operator $\ctrlZ{}$ defined thus for every computational basis state $\ket{x_1x_2\cdots x_k}$, for $x_1,x_2,\ldots,x_k\in\{0,1\}$:
\[ \ctrlZ{}\ket{x_1x_2\cdots x_k} = (-1)^{x_1\cdots x_k}\ket{x_1x_2\cdots x_k}\;. \]
Generalizing this definition a bit, for any $\eta\in\complexes$ such that $|\eta| = 1$ and $\eta \ne 1$, we define the unitary gate $G_\eta$ by
\[ G_\eta\ket{x_1x_2\cdots x_k} = \eta^{x_1\cdots x_k}\ket{x_1x_2\cdots x_k}\;. \]
$G_\eta$ is represented in the computational basis by a diagonal matrix, and it has two eigenspaces: the one-dimensional subspace $E :=\{a\ket{1\cdots 1} : a\in\complexes\}$ with eigenvalue $\eta$, and its orthogonal complement $E^\perp$ with eigenvalue $1$.  $E^\perp$ is spanned by those basis vectors with at least one $0$ in their corresponding strings.  Note that $G_\eta$ is unitary, that $G_\eta^* = G_{\eta^*}$, and that $G_\eta$ commutes with the swap operator on any pair of its qubits.

We now fix $\cH := \cH_n$ to be the $n$-qubit Hilbert space, for some $n > 0$.  We let the qubits of $\cH$ have indices from $1$ to $n$.  The computational basis of $\cH$ is thus $\left\{ \ket{x} \mid \map{x}{[n]}{\{0,1\}}\right\}$, indexed by binary strings of length $n$.  Recall that for any fixed subset $S\subseteq [n]$, we let $\cH_S$ denote the Hilbert space of the qubits in $S$ (or, more strictly speaking, the qubits whose indices are in $S$).  So for example, $\cH \cong \cH_S \otimes \cH_{\overline S}$, where we write $\overline S$ for $[n]\setminus S$.  Similarly, if $\map{x}{[n]}{\{0,1\}}$ is any length-$n$ binary string, we let $x_{|S}$ denote the restriction of $x$ to $S$, and for $i\in [n]$ we write $x_i$ for $x_{|\{i\}}$.  We use the term, ``string'' to refer generally to $0,1$-valued maps whose domains are arbitrary subsets of $[n]$.  If we do not specify the domain of a string, we assume it is $[n]$.

We let $\1$ denote the string of $n$ many $1$'s, i.e., the constant $1$-valued string with domain $[n]$.

If $\map{y}{J}{\{0,1\}}$ and $\map{z}{K}{\{0,1\}}$ are strings for disjoint sets $J,K\subseteq [n]$, then we write $y\cup z$ for the unique string with domain $J\cup K$ extending $y$ and $z$.  So in particular, for any $S\subseteq [n]$, if $\ket{y}$ is a computational basis state of $\cH_S$ and $\ket{z}$ is a computational basis state of $\cH_{\overline S}$, then $\ket{y\cup z}$ is the computational basis state of $\cH$ corresponding to $\ket{y}\otimes\ket{z}$.


%

We now fix for the entire sequel some arbitrary $\eta\in\complexes$ such that $|\eta|=1$ and $\eta\ne 1$.

\begin{defn}
For any set $S\subseteq [n]$, let $G_S$ be the $G_\eta$ gate applied to the qubits in $S$ (which means that $G_S$ is an operator on $\cH_S$).  Similarly, let $I_S$ be the identity operator applied to the qubits in $S$.  (If $S=\emptyset$, then $\cH_S$ has dimension~$1$ and we define $G_S := \eta I_S$ by convention.)
\end{defn}

The next definition essentially repeats Definition~\ref{def:simplify} but for $G_\eta$ gates.

\begin{defn}\label{def:simplify-2}
Let $S\subseteq [n]$ be any set, and let $G := G_S\otimes I_{\overline{S}}$ (so $G$ is an operator on $\cH$).  Let $\ket{\psi}\in\cH$ be some unit vector.    We say that $G$ \emph{simplifies} on $\ket{\psi}$ if either (a) $G\ket{\psi} = \ket{\psi}$ or (b) $G\ket{\psi} = (G_T\otimes I_{\overline{T}})\ket{\psi} \ne \ket{\psi}$ for some proper subset $T\subset S$.  In case (a), we say that \emph{$G$ disappears (or is turned off) on $\ket{\psi}$}; in case (b), we say that \emph{$G$ simplifies to $T$ on $\ket{\psi}$}.
\end{defn}

As with Definition~\ref{def:simplify}, there are two ways that $G$ can simplify on $\ket{\psi}$: either (case~(a)) $\braket{x}{\psi} = 0$ for every string $x$ such that $x_{|S} = \1_{|S}$ (whence $G\ket{\psi} = \ket{\psi}$), or (case~(b)) there exists $i\in S$ such that, for all strings $x$ with $x_i = 0$, we have $\braket{x}{\psi} = 0$.  In case~(a), $\ket{\psi}$ is an eigenvector of $G$ with eigenvalue $1$ and so $G\ket{\psi} = \ket{\psi}$; every computational basis vector appearing in the expansion of $\ket{\psi}$ (as a linear combination of computational basis vectors) has a $0$ somewhere in $S$.  These $0$'s turn off $G$.  In case~(b), $G\ket{\psi} = (G_{S\setminus\{i\}} \otimes I_{\overline S \cup \{i\}})\ket{\psi}$, that is, $G$ acts the same as a smaller $G_\eta$ gate applied to all qubits in $S$ except the $i^\text{th}$.  This can only happen if $\ket{\psi} = \ket{1}_{\{i\}} \otimes \ket{\psi'}$, where $\ket{1}_{\{i\}}\in\cH_{\{i\}}$ and $\ket{\psi'}$ is some state in $\cH_{\overline{\{i\}}}$.

We now restate the entanglement lemma in this more general context.

\begin{lemma}\label{lem:main}
Let $S\subseteq [n]$ be arbitrary, and let $G := G_S\otimes I_{\overline{S}}$.  Let $\ket{\psi}\in\cH$ be any unit vector.  Then at least one of the following is true: (1) $\ket{\psi}$ is $S$-entangled; (2) $G\ket{\psi}$ is $S$-entangled; or (3) $G$ simplifies on $\ket{\psi}$.
\end{lemma}

\begin{proof}
The case where $|S|\le 1$ is trivial (every state is $S$-entangled), so we assume that $|S|\ge 2$.  Let $\ket{\p} := G\ket{\psi}$.
Since $G$ is represented by a diagonal matrix, for any string $x$, we have $|\braket{x}{\p}| = |\matelt{x}{G}{\psi}| = |\braket{x}{\psi}|$, so in particular, $\braket{x}{\p} = 0$ if and only if $\braket{x}{\psi} = 0$.

Suppose $\ket{\psi}$ and $\ket{\p}$ are both $S$-separable.  Write $\ket{\psi} = \ket{\psi}_A \otimes \ket{\psi}_B$, where $A\disjunion B = [n]$, $A$ and $B$ each have nonempty intersection with $S$, and $\ket{\psi}_A\in\cH_A$ and $\ket{\psi}_B\in\cH_B$ are unit vectors.  Likewise, write $\ket{\p} = \ket{\p}_C \otimes \ket{\p}_D$, for $C$ and $D$ where $C\disjunion D = [n]$, each have nonempty intersection with $S$, and $\ket{\p}_C\in\cH_C$ and $\ket{\p}_D\in\cH_D$ are unit vectors.

Now assume for the sake of contradiction that $G$ does not simplify on $\ket{\psi}$.  Then we have $G\ket{\psi} \ne \ket{\psi}$, and so there exists a string $u$ such that $u_{|S} = \1_{|S}$ and $\braket{u}{\psi} \ne 0$.  Fix such a $u$, noting that $G\ket{u} = \eta\,\ket{u}$.


We say that a string $\map{x}{[n]}{\{0,1\}}$ is a \emph{test string} if, for every nonempty $Y \in \{S\cap A\cap C, S\cap A\cap D, S\cap B\cap C, S\cap B\cap D\}$, there exists $i\in Y$ such that $x_i = 0$.  We will derive a contradiction in two steps: (1) show that $\braket{x}{\psi} = 0$ for every test string $x$; and (2) construct a test string $y$ such that $\braket{y}{\psi} \ne 0$.

To show step~(1), fix an arbitrary test string $x$.  We first chop $x$ into two parts in two different ways: (1) $x_{|A}$ and $x_{|B}$; (2) $x_{|C}$ and $x_{|D}$.  Each pair unions to $x$.  From $x_{|A}$ we get four strings $\map{x^A_{jk}}{A}{\{0,1\}}$ for $j,k\in\{0,1\}$ by changing some $0$-entries in $x_{|A}$ to $1$: Define
\begin{align*}
x^A_{00} &:= x_{|A}\;, &
x^A_{01} &:= x_{|A\cap C} \cup u_{|A\cap D}\;, \\
x^A_{10} &:= u_{|A\cap C} \cup x_{|A\cap D}\;, &
x^A_{11} &:= u_{|A}\;.
\end{align*}
We make similar definitions using $x_{|B}$, $x_{|C}$, and $x_{|D}$ with domains $B$, $C$, and $D$, respectively: Define
\begin{align*}
x^B_{00} &:= x_{|B}\;, &
x^B_{01} &:= x_{|B\cap C} \cup u_{|B\cap D}\;, \\
x^B_{10} &:= u_{|B\cap C} \cup x_{|B\cap D}\;, &
x^B_{11} &:= u_{|B}\;, \\ \\
x^C_{00} &:= x_{|C}\;, &
x^C_{01} &:= x_{|C\cap A} \cup u_{|C\cap B}\;, \\
x^C_{10} &:= u_{|C\cap A} \cup x_{|C\cap B}\;, &
x^C_{11} &:= u_{|C}\;, \\ \\
x^D_{00} &:= x_{|D}\;, &
x^D_{01} &:= x_{|D\cap A} \cup \1_{|D\cap B}\;, \\
x^D_{10} &:= u_{|D\cap A} \cup x_{|D\cap B}\;, &
x^D_{11} &:= u_{|D}\;.
\end{align*}
There are two things to observe about these definitions:
\begin{enumerate}
\item
We have $x^A_{00}\cup x^B_{00} = x^C_{00}\cup x^D_{00} = x$.
\item
For all $j,k,\ell,m\in\{0,1\}$,
\begin{equation}\label{eqn:ABCD-strings}
x^A_{jk}\cup x^B_{\ell m} = x^C_{j\ell}\cup x^D_{km}\;.
\end{equation}
For example, for all $i\in [n]$, we have
\[ (x^A_{00}\cup x^B_{10})_i = (x^C_{01}\cup x^D_{00})_i = \left\{ \begin{array}{ll}
u_i & \mbox{if $i \in B\cap C$,} \\
x_i & \mbox{otherwise.}
\end{array} \right. \]
\end{enumerate}

We now consider only the coefficients in $\ket{\psi}_A$, $\ket{\psi}_B$, $\ket{\psi}_C$, and $\ket{\psi}_D$ of the basis vectors given above.  For all $j,k\in\{0,1\}$, define
\begin{align*}
a_{jk} &:= \braket{x^A_{jk}}{\psi}_A &\mbox{(scalar product in $\cH_A$),} \\
b_{jk} &:= \braket{x^B_{jk}}{\psi}_B &\mbox{(scalar product in $\cH_B$),} \\
c_{jk} &:= \braket{x^C_{jk}}{\p}_C &\mbox{(scalar product in $\cH_C$),} \\
d_{jk} &:= \braket{x^D_{jk}}{\p}_D &\mbox{(scalar product in $\cH_D$).}
\end{align*}
For example, $a_{jk}$ is the coefficient of $\ket{x^A_{jk}}$ in the expansion of $\ket{\psi}_A$ in terms of basis vectors in $\cH_A$.  (The $a_{jk}$, $b_{jk}$, $c_{jk}$, and $d_{jk}$ may depend on the particular choice of test string $x$.)

Recalling that $\ket{\psi} = \ket{\psi}_A \otimes \ket{\psi}_B$ and $\ket{\p} = \ket{\p}_C \otimes \ket{\p}_D$, we get, for all $j,k,\ell,m\in\{0,1\}$,
\begin{align}\label{eqn:case-1-AB}
\braket{x^A_{jk}\cup x^B_{\ell m}}{\psi} &= (\bra{x^A_{jk}}\otimes\bra{x^B_{\ell m}})(\ket{\psi}_A \otimes \ket{\psi}_B) = \braket{x^A_{jk}}{\psi}_A\braket{x^B_{\ell m}}{\psi}_B = a_{jk}b_{\ell m}\;, \\ \label{eqn:case-1-CD}
\braket{x^C_{jk}\cup x^D_{\ell m}}{\p} &= (\bra{x^C_{jk}}\otimes\bra{x^D_{\ell m}})(\ket{\p}_C \otimes \ket{\p}_D) = \braket{x^C_{jk}}{\p}_A\braket{x^D_{\ell m}}{\p}_B = c_{jk}d_{\ell m}\;.
\end{align}
Then by observation~(1), if we can show that $a_{00}b_{00} = 0$, then $\braket{x}{\psi} = a_{00}b_{00} = 0$ for any test string $x$.

For any string $\map{x}{[n]}{\{0,1\}}$, if there exists $i\in S$ such that $x_i = 0$, then $G\ket{x} = \ket{x}$, and if $x_i = 1$ for all $i\in S$, then $G\ket{x} = \eta\,\ket{x}$.  This fact gives us equations among the $a_{jk},b_{jk},c_{jk},d_{jk}$ by comparing amplitudes in $\ket{\psi}$ versus $\ket{\p}$.  Which equations we get depends on which of the sets $S\cap A\cap C$, $S\cap A\cap D$, $S\cap B\cap C$, and $S\cap B\cap D$ are empty.  At most two of these sets can be empty, so we have three cases.

\paragraph{Case~1.} $S\cap A\cap C$, $S\cap A\cap D$, $S\cap B\cap C$, and $S\cap B\cap D$ are all nonempty.

In this case, $x^C_{11}\cup x^D_{11} = u$, and if $jk\ell m = 0$ then $x^A_{j\ell}\cup x^D_{km}$ has a $0$ somewhere in $S$.  Using observation~(2) above, we then get
\begin{equation}\label{eqn:case-1-G-on-basis}
G\ket{x^A_{jk}\cup x^B_{\ell m}} = G\ket{x^C_{j\ell}\cup x^D_{km}} = \left\{ \begin{array}{ll}
\eta\,\ket{x^C_{j\ell}\cup x^D_{km}} & \mbox{if $j=k=\ell=m=1$,} \\
\ket{x^C_{j\ell}\cup x^D_{km}} & \mbox{if $jk\ell m = 0$.}
\end{array}\right.
\end{equation}
Then combining Equations~(\ref{eqn:ABCD-strings},\ref{eqn:case-1-AB},\ref{eqn:case-1-CD},\ref{eqn:case-1-G-on-basis}) and the fact that $\ket{\p} = G\ket{\psi}$, we get 16 equations: for all $j,k,\ell,m\in\{0,1\}$,
\begin{equation}\label{eqn:case-1}
c_{j\ell}d_{km} = \left\{ \begin{array}{ll}
\eta\,a_{jk}b_{\ell m} & \mbox{if $j=k=\ell=m=1$,} \\
a_{jk}b_{\ell m} & \mbox{if $jk\ell m = 0$.}
\end{array}\right.
\end{equation}

By assumption, $\braket{u}{\psi} \ne 0$, and so $\braket{u}{\psi} = a_{11}b_{11} \ne 0$, and $c_{11}d_{11} = \eta\,a_{11}b_{11} \ne 0$ as well.  This fact together with Equation~(\ref{eqn:case-1}) implies
$a_{00}b_{00} = 0$ by Lemma~\ref{lem:4-sets} in Appendix~\ref{sec:calculations}.

\paragraph{Case~2.}  One of $S\cap A\cap C$, $S\cap A\cap D$, $S\cap B\cap C$, and $S\cap B\cap D$ is empty and the other three are nonempty.  Without loss of generality, we assume that $S\cap B\cap C = \emptyset$.

In this case, $x^C_{j\ell}\cup x^D_{km} = u$ if $j = k = m = 1$ (independent of $\ell$, because the test string $x$ has no $0$ in $S\cap B\cap C$), and otherwise if $jkm = 0$, we get that $x^A_{j\ell}\cup x^D_{km}$ has a $0$ somewhere in $S$.  Thus
\begin{equation}\label{eqn:case-2-G-on-basis}
G\ket{x^A_{jk}\cup x^B_{\ell m}} = G\ket{x^C_{j\ell}\cup x^D_{km}} = \left\{ \begin{array}{ll}
\eta\,\ket{x^C_{j\ell}\cup x^D_{km}} & \mbox{if $j=k=m=1$,} \\
\ket{x^C_{j\ell}\cup x^D_{km}} & \mbox{if $jkm = 0$.}
\end{array}\right.
\end{equation}
Then setting $\ell := 0$ we get eight equations: for all $j,k,m\in\{0,1\}$,
\begin{equation}\label{eqn:case-2}
c_{j0}d_{km} = \left\{ \begin{array}{ll}
\eta\,a_{jk}b_{0m} & \mbox{if $j=k=m=1$,} \\
a_{jk}b_{0m} & \mbox{if $jkm = 0$.}
\end{array}\right.
\end{equation}
These equations again imply $a_{00}b_{00} = 0$ by Lemma~\ref{lem:3-sets} in Appendix~\ref{sec:calculations}.

\paragraph{Case~3.}  Two of $S\cap A\cap C$, $S\cap A\cap D$, $S\cap B\cap C$, and $S\cap B\cap D$ are empty.  Without loss of generality, we assume that $S\cap A\cap D = S\cap B\cap C = \emptyset$, whence $S\cap A = S\cap C$ and $S\cap B = S\cap D$, and both are nonempty.  We argue analogously to Cases~1 and 2.

In this case, $x^C_{j\ell}\cup x^D_{km} = u$ if $j = m = 1$ (independent of $k$ and $\ell$), and otherwise if $jm = 0$, we get that $x^A_{j\ell}\cup x^D_{km}$ has a $0$ somewhere in $S$.  Thus
\begin{equation}\label{eqn:case-3-G-on-basis}
G\ket{x^A_{jk}\cup x^B_{\ell m}} = G\ket{x^C_{j\ell}\cup x^D_{km}} = \left\{ \begin{array}{ll}
\eta\,\ket{x^C_{j\ell}\cup x^D_{km}} & \mbox{if $j=m=1$,} \\
\ket{x^C_{j\ell}\cup x^D_{km}} & \mbox{if $jm = 0$.}
\end{array}\right.
\end{equation}
Then setting $k := \ell := 0$ we get four equations: for all $j,m\in\{0,1\}$,
\begin{equation}\label{eqn:case-3}
c_{j0}d_{0m} = \left\{ \begin{array}{ll}
\eta\,a_{j0}b_{0m} & \mbox{if $j=m=1$,} \\
a_{j0}b_{0m} & \mbox{if $jm = 0$.}
\end{array}\right.
\end{equation}
These equations also imply $a_{00}b_{00} = 0$ by Lemma~\ref{lem:2-sets} in Appendix~\ref{sec:calculations}.

This establishes step~(1) in the contradiction proof.

For step~(2), we now construct a test string $y$ such that $\braket{y}{\psi} \ne 0$.  We first show the construction assuming Case~1 above, then modify it slightly for Cases~2 and 3.

Assume Case~1.  Choose some $i\in S\cap A\cap C$.  Since $G$ does not simplify on $\ket{\psi}$, there exists a string $y_{AC}$ (with domain $[n]$) such that $\braket{y_{AC}}{\psi} \ne 0$ and $(y_{AC})_i = 0$.  Then since $G$ fixes $\ket{y_{AC}}$, we have
\[ 0 \ne \braket{y_{AC}}{\psi} = \braket{y_{AC}}{\p} = \braket{(y_{AC})_{|C}\cup (y_{AC})_{|D}}{\p} = \braket{(y_{AC})_{|C}}{\p}_C\braket{(y_{AC})_{|D}}{\p}_D\;. \]
In particular, $\braket{(y_{AC})_{|C}}{\p}_C \ne 0$.  Now we can choose some string $y_{AD}$ such that $(y_{AD})_i = 0$ for some $i\in S\cap A\cap D$.  Analogously to the above, we get
\[ 0 \ne \braket{y_{AD}}{\psi} = \braket{y_{AD}}{\p} = \braket{(y_{AD})_{|C}
\cup(y_{AD})_{|D}}{\p} = \braket{(y_{AD})_{|C}}{\p}_C\braket{(y_{AD})_{|D}}{\p}_D\;. \]
In particular, $\braket{(y_{AD})_{|D}}{\p}_D \ne 0$.  Now define the string
\[ y_A := (y_{AC})_{|C} \cup (y_{AD})_{|D}\;. \]
Note that $(y_A)_i = (y_A)_j = 0$ for some $i\in S\cap A\cap C$ and $j\in S\cap A\cap D$.  Furthermore,
\[ \braket{y_A}{\psi} = \braket{y_A}{\p} = \braket{(y_{AC})_{|C}}{\p}_C \braket{(y_{AD})_{|D}}{\p}_D \ne 0\;. \]

By exactly repeating the argument in the previous paragraph with $B$ substituted for $A$, we obtain a string $y_B$ such that $(y_B)_i = (y_B)_j = 0$ for some $i\in S\cap B\cap C$ and $j\in S\cap B\cap D$, and furthermore, $\braket{y_B}{\psi} \ne 0$.

Finally, let $y := (y_A)_{|A} \cup (y_B)_{|B}$.  Observe that $y$ is a test string and that
\[ \braket{y}{\psi} = \braket{y_{|A}\cup y_{|B}}{\psi} = \braket{y_{|A}}{\psi}_A \braket{y_{|B}}{\psi}_B \ne 0\;. \]
This concludes the proof for Case~1.

Assume Case~2.  Using an identical construction to that of Case~1, we obtain a string $y_A$ such that $\braket{y_A}{\psi} \ne 0$ and $(y_A)_i = (y_A)_j = 0$ for some $i\in S\cap A\cap C$ and $j\in S\cap A\cap D$.  Let $y_B$ be any string such that $\braket{y_B}{\psi} \ne 0$ and $(y_B)_i = 0$ for some $i\in S\cap B$.  Such a string exists by the assumption that $G$ does not simplify on $\ket{\psi}$.  Now letting $y := (y_A)_{|A} \cup (y_B)_{|B}$ as in Case~1, we observe that $y$ is a test string and that
\[ \braket{y}{\psi} = \braket{y_{|A}\cup y_{|B}}{\psi} = \braket{y_{|A}}{\psi}_A \braket{y_{|B}}{\psi}_B \ne 0\;. \]
This concludes the proof of Case~2.

Assume Case~3.  Let $y_A$ be any string such that $\braket{y_A}{\psi} \ne 0$ and $(y_A)_i = 0$ for some $i\in S\cap A$.  Let $y_B$ be any string such that $\braket{y_B}{\psi} \ne 0$ and $(y_B)_i = 0$ for some $i\in S\cap B$.  Both strings exist by the assumption that $G$ does not simplify on $\ket{\psi}$.  Now letting $y := (y_A)_{|A} \cup (y_B)_{|B}$ as in Cases~1 and 2, we observe that $y$ is a test string and that
\[ \braket{y}{\psi} = \braket{y_{|A}\cup y_{|B}}{\psi} = \braket{y_{|A}}{\psi}_A \braket{y_{|B}}{\psi}_B \ne 0\;. \]
This concludes the proof of Case~3.
\end{proof}

%
%
%
%

\section{Calculations}
\label{sec:calculations}

\begin{lemma}\label{lem:4-sets}
Let $\eta\in\complexes$ be such that $\eta \ne 1$.  Let complex numbers $a_{jk}$, $b_{jk}$, $c_{jk}$, and $d_{jk}$ for $j,k\in\{0,1\}$ satisfy
\begin{align}\label{eqn:4-vals-all-1}
a_{11}b_{11} &= \eta\, c_{11}d_{11}\;, \\ \label{eqn:4-vals-not-all-1}
a_{jk}b_{\ell m} &= c_{j\ell}d_{km}
\end{align}
for all $j,k,\ell,m\in\{0,1\}$ such that $jk\ell m = 0$.  If $a_{11}$ and $b_{11}$ are nonzero, then either $c_{00} = c_{01} = c_{10} = 0$ or $d_{00} = d_{01} = d_{10} = 0$.  It follows that for all $r,s\in\{0,1\}$,
\begin{equation}\label{eqn:4-vals-rs}
a_{r0}b_{0s} = a_{0r}b_{s0} = c_{r0}d_{0s} = c_{0r}d_{s0} = 0\;.
\end{equation}
\end{lemma}

\begin{proof}
If $a_{11}b_{11} \ne 0$, then by Equation~(\ref{eqn:4-vals-all-1}) we have $\eta$, $c_{11}$, and $d_{11}$ are all nonzero as well.  Letting $j := k := 1$ in Equations~(\ref{eqn:4-vals-all-1},\ref{eqn:4-vals-not-all-1}), we can solve for each $b_{\ell m}$ in terms of the other quantities:
\begin{align*}
b_{00} &= c_{10}d_{10}/a_{11} & b_{01} &= c_{10}d_{11}/a_{11} \\
b_{10} &= c_{11}d_{10}/a_{11} & b_{11} &= \eta\,c_{11}d_{11}/a_{11}
\end{align*}
Substituting these values into the other 12 equations (where $jk = 0$) and simplifying, we get
\begin{align*}
a_{00}c_{10}d_{10} &= a_{11}c_{00}d_{00} & a_{01}c_{10}d_{10} &= a_{11}c_{00}d_{10} & a_{10}c_{10}d_{10} &= a_{11}c_{10}d_{00} \\
a_{00}c_{10}d_{11} &= a_{11}c_{00}d_{01} & a_{01}c_{10} &= a_{11}c_{00} & a_{10}c_{10}d_{11} &= a_{11}c_{10}d_{01} \\
a_{00}c_{11}d_{10} &= a_{11}c_{01}d_{00} & a_{01}c_{11}d_{10} &= a_{11}c_{01}d_{10} & a_{10}d_{10} &= a_{11}d_{00} \\
\eta\,a_{00}c_{11}d_{11} &= a_{11}c_{01}d_{01} & \eta\,a_{01}c_{11} &= a_{11}c_{01} & \eta\,a_{10}d_{11} &= a_{11}d_{01}
\end{align*}
Using the three equations on the bottom row, we solve for $a_{00}$, $a_{01}$, and $a_{10}$:
\begin{align*}
a_{00} &= \frac{a_{11}c_{01}d_{01}}{\eta\,c_{11}d_{11}} & a_{01} &= \frac{a_{11}c_{01}}{\eta\,c_{11}} & a_{10} &= \frac{a_{11}d_{01}}{\eta\,d_{11}}
\end{align*}
and plug these values into the remaining nine equations and simplify to get
\begin{align*}
c_{01}c_{10}d_{01}d_{10} &= \eta\,c_{00}c_{11}d_{00}d_{11} & c_{01}c_{10}d_{01} &= \eta\,c_{00}c_{11}d_{10} & c_{10}d_{01}d_{10} &= \eta\,c_{10}d_{00}d_{11} \\
c_{01}c_{10}d_{01} &= \eta\,c_{00}c_{11}d_{01} & c_{01}c_{10} &= \eta\,c_{00}c_{11} & c_{10}d_{01} &= \eta\,c_{10}d_{01} \\
c_{01}d_{01}d_{10} &= \eta\,c_{01}d_{00}d_{11} & c_{01}d_{10} &= \eta\,c_{01}d_{10} & d_{01}d_{10} &= \eta\,d_{00}d_{11}
\end{align*}
Noting that $\eta \ne 1$, from the last equation on the second row and the second equation on the last row we get
\[ c_{10}d_{01} = c_{01}d_{10} = 0\;. \]
Substituting these values into the equations on the first row and first column, we get for the seven remaining equations
\begin{align*}
c_{00}d_{00} &= 0 & c_{00}d_{10} &= 0 & c_{10}d_{00} &= 0 \\
c_{00}d_{01} &= 0 & c_{01}c_{10} &= \eta\,c_{00}c_{11} & & \\
c_{01}d_{00} &= 0 & & & d_{01}d_{10} &= \eta\,d_{00}d_{11}
\end{align*}
Suppose $c_{00} \ne 0$.  Then the top left equation and its two adjacent equations imply $d_{00} = d_{01} = d_{10} = 0$.  Symmetrically, if $d_{00} \ne 0$, then the corner equations imply $c_{00} = c_{01} = c_{10} = 0$.  Combining this fact with Equation~(\ref{eqn:4-vals-not-all-1}) gives us Equation~(\ref{eqn:4-vals-rs}).
\end{proof}

\begin{rmrk}
The proof above did not use the two equations $c_{01}c_{10} = \eta\,c_{00}c_{11}$ and $d_{01}d_{10} = \eta\,d_{00}d_{11}$.  They show that $c_{00}$ is uniquely determined by the other $c$'s and $\eta$.  Also, if $c_{00} \ne 0$, then $c_{01} \ne 0$ and $c_{10} \ne 0$, and conversely.  Similarly for the $d$'s.
\end{rmrk}

\begin{lemma}\label{lem:3-sets}
Let $\eta\in\complexes$ be such that $\eta \ne 1$.  Let complex numbers $a_{jk}$, $b_{j}$, $c_{j}$, and $d_{jk}$ for $j,k\in\{0,1\}$ satisfy
\begin{align}\label{eqn:3-vals-all-1}
a_{11}b_{1} &= \eta\, c_{1}d_{11}\;, \\ \label{eqn:3-vals-not-all-1}
a_{jk}b_{m} &= c_{j}d_{km}
\end{align}
for all $j,k,m\in\{0,1\}$ such that $jkm = 0$.  If $a_{11}$ and $b_{1}$ are nonzero, then either $c_{0} = 0$ or $d_{00} = d_{10} = 0$.  Thus
\begin{equation}\label{eqn:3-vals-rs}
a_{00}b_{0} = c_{0}d_{00} = a_{01}b_{0} = c_{0}d_{10} = 0\;.
\end{equation}
\end{lemma}

\begin{proof}
If $a_{11}b_{1} \ne 0$, then $\eta$, $c_{1}$, and $d_{11}$ are all nonzero as well.  Letting $j := k := 1$ in Equations~(\ref{eqn:3-vals-all-1},\ref{eqn:3-vals-not-all-1}), we can solve for each $b_{m}$ in terms of the other quantities:
\begin{align*}
b_{0} &= c_{1}d_{10}/a_{11} & b_{1} &= \eta\,c_{1}d_{11}/a_{11}
\end{align*}
Substituting these values into the other six equations (where $jk = 0$) and simplifying, we get
\begin{align*}
a_{00}c_{1}d_{10} &= a_{11}c_{0}d_{00} & a_{01}c_{1}d_{10} &= a_{11}c_{0}d_{10} & a_{10}d_{10} &= a_{11}d_{00} \\
\eta\,a_{00}c_{1}d_{11} &= a_{11}c_{0}d_{01} & \eta\,a_{01}c_{1} &= a_{11}c_{0} & \eta\,a_{10}d_{11} &= a_{11}d_{01}
\end{align*}
Using the three equations on the bottom row, we solve for $a_{00}$, $a_{01}$, and $a_{10}$:
\begin{align*}
a_{00} &= \frac{a_{11}c_{0}d_{01}}{\eta\,c_{1}d_{11}} & a_{01} &= \frac{a_{11}c_{0}}{\eta\,c_{1}} & a_{10} &= \frac{a_{11}d_{01}}{\eta\,d_{11}}
\end{align*}
and plug these values into the remaining three equations and simplify to get
\begin{align*}
c_{0}d_{01}d_{10} &= \eta\,c_{0}d_{00}d_{11} & c_{0}d_{10} &= \eta\,c_{0}d_{10} & d_{01}d_{10} &= \eta\,d_{00}d_{11}
\end{align*}
Noting that $\eta \ne 1$, from the middle equation we get that
\begin{equation}\label{eqn:c0d10}
c_{0}d_{10} = 0\;.
\end{equation}
Substituting these values into the first equation gives
\begin{align}\label{eqn:c0d00}
c_0d_{00} &= 0 & d_{01}d_{10} &= \eta\,d_{00}d_{11}
\end{align}
If $c_{0} \ne 0$, then $d_{00} = d_{10} = 0$ by (\ref{eqn:c0d10},\ref{eqn:c0d00}).  Combining this fact with Equation~(\ref{eqn:3-vals-not-all-1}) gives us Equation~(\ref{eqn:3-vals-rs}).
\end{proof}

\begin{rmrk}
The unused second equation of (\ref{eqn:c0d00}) shows that $d_{00}$ is uniquely determined by the other $d$'s and $\eta$.  Also, if $d_{00} \ne 0$, then $d_{01} \ne 0$ and $d_{10} \ne 0$, and conversely.
\end{rmrk}

\begin{lemma}\label{lem:2-sets}
Let $\eta\in\complexes$ be such that $\eta \ne 1$.  Let complex numbers $a_{j}$, $b_{j}$, $c_{j}$, and $d_{j}$ for $j\in\{0,1\}$ satisfy
\begin{align}\label{eqn:2-vals-all-1}
a_{1}b_{1} &= \eta\, c_{1}d_{1}\;, \\ \label{eqn:2-vals-not-all-1}
a_{j}b_{m} &= c_{j}d_{m}
\end{align}
for all $j,m\in\{0,1\}$ such that $jm = 0$.  If $a_{1}$ and $b_{1}$ are nonzero, then
\begin{equation}\label{eqn:2-vals-rs}
a_{0}b_{0} = c_{0}d_{0} = 0\;.
\end{equation}
\end{lemma}

\begin{proof}
If $a_{1}b_{1} \ne 0$, then $\eta$, $c_{1}$, and $d_{1}$ are all nonzero as well.  Letting $j := 1$ in Equations~(\ref{eqn:2-vals-all-1},\ref{eqn:2-vals-not-all-1}), we can solve for each $b_{m}$ in terms of the other quantities:
\begin{align*}
b_{0} &= c_{1}d_{0}/a_{1} & b_{1} &= \eta\,c_{1}d_{1}/a_{1}
\end{align*}
Substituting these values into the other two equations (where $j = 0$) and simplifying, we get
\begin{align*}
a_{0}c_{1}d_{0} &= a_{1}c_{0}d_{0} & \eta\,a_{0}c_{1} &= a_{1}c_{0}
\end{align*}
We use the second equation to solve for $a_{0}$:
\begin{align*}
a_{0} &= \frac{a_{1}c_{0}}{\eta\,c_{1}}
\end{align*}
and plug this value into the first equation and simplify to get
\begin{align*}
c_{0}d_{0} &= \eta\,c_{0}d_{0}
\end{align*}
Noting that $\eta \ne 1$, we get that
\begin{equation}\label{eqn:c0d0}
c_{0}d_{0} = 0\;.
\end{equation}
Combining Equations~(\ref{eqn:2-vals-not-all-1},\ref{eqn:c0d0}) gives us Equation~(\ref{eqn:2-vals-rs}).
\end{proof}

\end{document}